\documentclass[final,leqno]{siamltex}

\usepackage{mathbf-abbrevs}

\newcommand{\reals}{{\mathbb R}}
\newcommand{\ints}{{\mathbb Z}}

\newcommand{\vor}{\operatorname{Vor}}
\newcommand{\relevant}{\operatorname{Rel}}

\newcommand{\term}{\emph}

\newcommand{\abs}[1]{{\left| #1 \right|}}
\newcommand{\floor}[1]{{\left\lfloor #1 \right\rfloor}}

\newcommand{\ceil}[1]{\lceil #1 \rceil}

\newcommand{\dotprod}[2]{ #1 \cdot #2}

\usepackage[square,comma,numbers,sort]{natbib}

\usepackage{units}
		
\usepackage{booktabs}
		
\usepackage{ifpdf}
\ifpdf
  \usepackage[pdftex]{graphicx}
  \usepackage[naturalnames]{hyperref}
\else
	\usepackage{graphicx}
\fi

\usepackage{amsmath,amsfonts,amssymb} 

\usepackage{mathrsfs}

\title{Finding a closest point in a lattice of Voronoi's first kind}

\author{Robby~G.~McKilliam, Alex Grant and I. Vaughan L. Clarkson}

\begin{document}

\maketitle

 \begin{abstract} 
We show that for those lattices of Voronoi's first kind with known obtuse superbasis, a closest lattice point can be computed in $O(n^4)$ operations where $n$ is the dimension of the lattice.  To achieve this a series of relevant lattice vectors that converges to a closest lattice point is found.  We show that the series converges after at most $n$ terms.  Each vector in the series can be efficiently computed in $O(n^3)$ operations using an algorithm to compute a minimum cut in an undirected flow network.  
\end{abstract}


\pagestyle{myheadings}
\thispagestyle{plain} 
\markboth{Finding a closest point in a lattice of Voronoi's first kind}{DRAFT \today}

\section{Introduction}\label{sec:introduction}

An $n$-dimensional \term{lattice} $\Lambda$ is a discrete set of vectors from $\reals^m$, $m \geq n$, formed by the integer linear combinations of a set of linearly independent basis vectors $b_1, \dots, b_n$ from $\reals^m$~\cite{SPLAG}.  That is, $\Lambda$ consists of all those vectors, or \emph{lattice points}, $x \in \reals^m$ satisfying
\[
  x = b_1 u_1 + b_2u_2 + \dots + b_n u_n \qquad u_1, \dots , u_n \in \ints. 
\] 
Given a lattice $\Lambda$ in $\reals^m$ and a vector $y \in \reals^m$, a problem of interest is to find a lattice point $x \in \Lambda$ such that the squared Euclidean norm
\[
\| y - x \|^2 = \sum_{i=1}^m (y_i - x_i)^2
\] 
is minimised.  This is called the \emph{closest lattice point problem} (or \emph{closest vector problem}) and a solution is called a \emph{closest lattice point} (or simply \emph{closest point}) to $y$. 
A related problem is to find a lattice point of minimum nonzero Euclidean length, that is, a lattice point of length
\[
\min_{x\in \Lambda \backslash \{ \zerobf \} } \| x \|^2,
\]
where $\Lambda \backslash  \{\zerobf\}$ denotes the set of lattice points not equal to the origin $\zerobf$.  This is called the \emph{shortest vector problem}.

The closest lattice point problem and the shortest vector problem have interested mathematicians and computer scientists due to their relationship with integer programming~\cite{Lenstra_integerprogramming1983,Kannan1987_fast_general_np,Babai1986}, the factoring of polynomials~\cite{Lenstra1982}, and cryptanalysis~\cite{Joux_toolbox_cryptanal1998,NyguyenStern_two_faces_crypto,Micciancio_lattice_based_post_quantum_crypto}.  
Solutions of the closest lattice point problem have engineering applications.  For example, if a lattice is used as a vector quantiser then the closest lattice point corresponds to the minimum distortion point~\cite{Conway1983VoronoiCodes,Conway1982VoronoiRegions,Conway1982FastQuantDec}.  If the lattice is used as a code, then the closest lattice point corresponds to what is called \emph{lattice decoding}\index{lattice decoding} and has been shown to yield arbitrarily good codes~\cite{Erex2004_lattice_decoding,Erez2005}.  The closest lattice point problem also occurs in communications systems involving multiple antennas~\cite{Ryan2008,Wubben_2011}.  The unwrapping of phase data for location estimation can also be posed as a closest lattice point problem and this has been applied to the global positioning system~\cite{Teunissen_GPS_1995,Hassibi_GPS_1998}.  The problem has also found applications to circular statistics~\cite{McKilliam_mean_dir_est_sq_arc_length2010}, single frequency estimation~\cite{McKilliamFrequencyEstimationByPhaseUnwrapping2009}, and related signal processing problems~\cite{McKilliam2007,Clarkson2007,McKilliam2009IndentifiabliltyAliasingPolyphase,Quinn_sparse_noisy_SSP_2012}.

The closest lattice point problem is known to be NP-hard under certain conditions when the lattice itself, or rather a basis thereof, is considered as an additional input parameter~\cite{micciancio_hardness_2001, Dinur2003_approximating_CVP_NP_hard, Jalden2005_sphere_decoding_complexity}. Nevertheless, algorithms exist that can compute a closest lattice point in reasonable time if the dimension is small~\cite{Pohst_sphere_decoder_1981,Kannan1987_fast_general_np,Agrell2002}.  These algorithms all require a number of operations that grows as $O(n^{O(n)})$ where $n$ is the dimension of the lattice.  Recently, Micciancio~\cite{Micciancio09adeterministic} described a solution for the closest lattice point problem that requires a number of operations that grows as $O(2^{2n})$.  This single exponential growth in complexity is the best known. 

Although the problem is NP-hard in general, fast algorithms are known for specific highly regular lattices, such as the root lattices $A_n$ and $D_n$, their dual lattices $A_n^*$ and $D_n^*$, the integer lattice $\ints^{n}$, and the Leech lattice~\cite[Chap. 4]{SPLAG}\cite{Conway1982FastQuantDec,Conway1986SoftDecLeechGolay, Clarkson1999:Anstar, McKilliam2008, McKilliam2008b, McKilliam2009CoxeterLattices,Vardy1993_leech_lattice_MLD}.  In this paper we consider a particular class of lattices, those of \emph{Voronoi's first kind}~\cite{ConwaySloane1992_voronoi_lattice_3d_obtuse_superbases,Valentin2003_coverings_tilings_low_dimension,Voronoi1908_main_paper}.  Each lattice of Voronoi's first kind has what is called an \emph{obtuse superbasis}.  We show that if the obtuse superbasis is known, then a closest lattice point can be computed in $O(n^4)$ operations.  This is achieved by enumerating a series of \emph{relevant vectors} of the lattice.  Each relevant vector in the series can be computed in $O(n^3)$ operations using an algorithm for computing a minimum cut in an undirected flow network~\cite{Picard_min_cuts_1974,Sankaran_solving_CDMA_mincut_1998,Ulukus_cdma_mincut_1998,Cormen2001}.  We show that the series converges to a closest lattice point after at most $n$ terms, resulting in $O(n^4)$ operations in total.  This result extends upon a recent result by some of the authors showing that a short vector in a lattice of Voronoi's first kind can be found by computing a minimum cut in a weighted graph~\cite{McKilliam_short_vectors_first_type_isit_2012}.

The paper is structured as follows.  Section~\ref{sec:voron-cells-relev} describes the relevant vectors and the \emph{Voronoi cell} of a lattice. Section~\ref{sec:iterative-slicer} describes a procedure to find a closest lattice point by enumerating a series of relevant vectors.  The series is guaranteed to converge to a closest lattice point after a finite number of terms.  In general the procedure might be computationally expensive because the number of terms required might be large and because computation of each relevant vector in the series might be expensive.  Section~\ref{sec:latt-voron-first} describes lattices of Voronoi's first kind and their obtuse superbasis.  In Section~\ref{sec:seri-relev-vect} it is shown that for these lattices, the series of relevant vectors results in a closest lattice point after at most $n$ terms.  Section~\ref{sec:comp-clos-relev} shows that each relevant vector in the series can be computed in $O(n^3)$ operations by computing a minimum cut in an undirected flow network.  Section~\ref{sec:discussion} discusses some potential applications of this algorithm and poses some interesting questions for future research.

\section{Voronoi cells and relevant vectors}\label{sec:voron-cells-relev}
\newcommand{\calR}{\mathcal{R}}
The (closed) \term{Voronoi cell}, denoted $\vor(\Lambda)$, of a lattice $\Lambda$ in $\reals^m$ is the subset of $\reals^m$ containing all points closer or of equal distance (here with respect to the Euclidean norm) to the lattice point at the origin than to any other lattice point. The Voronoi cell is an $m$-dimensional convex polytope that is symmetric about the origin. 
 
Equivalently the Voronoi cell can be defined as the intersection of the half spaces 
\begin{align*}
H_{v} &= \{x \in \reals^n \mid \|x\| \leq \|x - v\| \} \\
&= \{x \in \reals^n \mid \dotprod{x}{v} \leq \tfrac{1}{2}\dotprod{v}{v} \}
\end{align*}
for all $v \in \Lambda \backslash  \{\zerobf\}$.  
We denote by $\dotprod{x}{v}$ the inner product between vectors $x$ and $v$.
It is not necessary to consider all $v \in \Lambda \backslash  \{\zerobf\}$ to define the Voronoi cell.   
The \emph{relevant vectors} are those lattice points $v \in \Lambda \backslash  \{\zerobf\}$ for which  
\[
\dotprod{v}{x} \leq \dotprod{x}{x} \qquad \text{for all $x \in \Lambda$}.
\]
We denote by $\relevant(\Lambda)$ the set of relevant vectors of the lattice $\Lambda$.  The Voronoi cell is the intersection of the halfspaces corresponding with the relevant vectors, that is, 
\[
\vor(\Lambda) = \cap_{v\in\relevant(\Lambda)}{H_{v}}.
\]
The closest lattice point problem and the Voronoi cell are related in that $x\in\Lambda$ is a closest lattice point to $y$ if and only $y - x \in \vor(\Lambda)$, that is, if and only if
\begin{equation}\label{eq:relvectnearpointieq}
\dotprod{(y - x)}{v} \leq \tfrac{1}{2} \dotprod{v}{v} \qquad \text{for all $v \in \relevant(\Lambda)$}.  
\end{equation} 


If $s$ is a short vector in a lattice $\Lambda$ then 
\[
\rho = \frac{\|s\|}{2} = \frac{1}{2} \min_{x \in \Lambda / \{\zerobf\} } \|x\|
\]
is called the \emph{packing radius} (or \emph{inradius}) of $\Lambda$~\cite{SPLAG}.  The packing radius is the minimum distance between the boundary of the Voronoi cell and the origin.  It is also the radius of the largest sphere that can be placed at every lattice point such that no two spheres intersect (see Figure~\ref{lattices:fig:vorregion}).  The following well known results will be useful.


\begin{proposition}\label{eq:latticepointsinvorcvell}
Let $\Lambda \subset \reals^{m}$ be an $n$-dimensional lattice.  For $r\in \reals$ let $\ceil{r}$ denote the smallest integer strictly larger than $r$.  Let $t \in \reals^m$.  The number of lattice points inside the scaled and translated Voronoi cell $r \vor(\Lambda) + t$ is at most $\ceil{r}^n$.
\end{proposition}
\begin{proof}
It is convenient to modify the boundary of the Voronoi cell so that it tessellates $\reals^m$ under translations by $\Lambda$.  With this aim we let $V \subset \vor(\Lambda)$ contain all those points from the interior of $\vor(\Lambda)$, but have opposing faces open and closed.  That is, if $x \in V$ is on the boundary of $V$ then $-x \notin V$.  With this definition it can be asserted that $V$ tessellates $\reals^m$ under translations by $\Lambda$.  That is, $\reals^{m} = \cup_{x \in \Lambda}(V + x)$ and the intersection $(V + x)\cap(V+y)$ is empty for distinct lattice points $x$ and $y$.  Now, for positive integer $k$, the scaled and translated cell $kV + t$ contains precisely one coset representative for each element of the quotient group $\Lambda/k\Lambda$~\cite[Sec.~2.4]{McKilliam2010thesis}.  There are $k^n$ coset representatives.  Thus, the number of lattice points inside $r \vor(\Lambda) + t \subset \ceil{r}V + t$ is at most $\ceil{r}^n$.
\end{proof}

\begin{proposition}\label{eq:latticepointsinsphere}
Let $\Lambda \subset \reals^{m}$ be an $n$-dimensional lattice with packing radius $\rho$.  Let $S$ be an $m$-dimensional hypersphere of radius $r$ centered at $t \in \reals^{m}$.  The number of lattice points from $\Lambda$ in the sphere $S$ is at most $\ceil{r/\rho}^{n}$. 
 \end{proposition}
 \begin{proof}
The packing radius $\rho$ is the Euclidean length of a point on the boundary of the Voronoi cell $\vor(\Lambda)$ that is closest to the origin. Therefore, the sphere $S$ is a subset of $\vor(\Lambda)$ scaled by $r/\rho$ and translated by $t$.  That is, $S \subset r/\rho \vor(\Lambda) + t$.  The proof follows because the number of lattice points in $r/\rho \vor(\Lambda) + t$ is at most $\ceil{r/\rho}^n$ by Proposition~\ref{eq:latticepointsinvorcvell}.
\end{proof}

 
\begin{figure}[tp] 
	\centering      
		\includegraphics{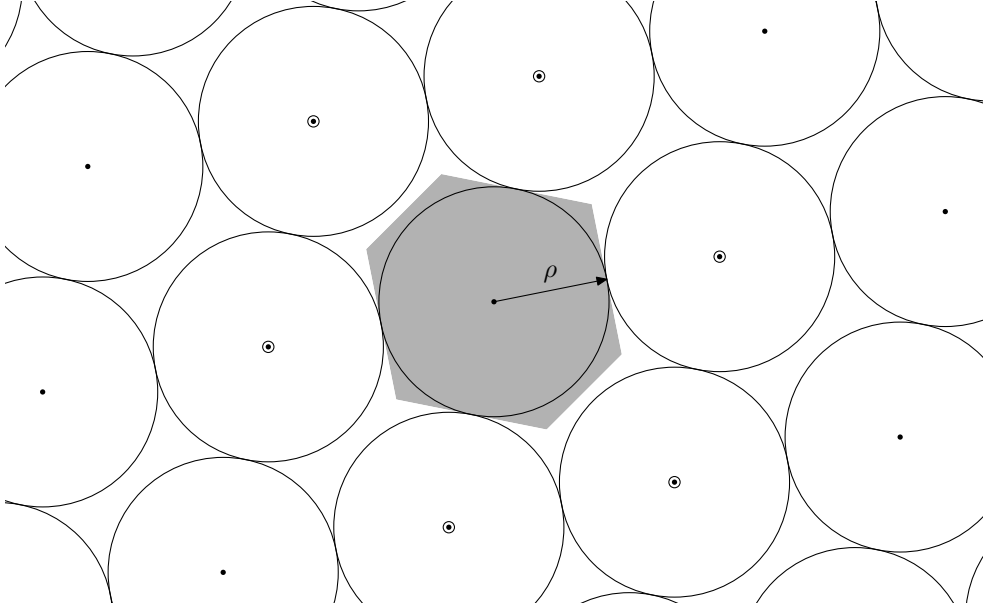} 
		\caption{The $2$-dimensional lattice with basis vectors $(3,0.6)$ and $(0.6,3)$.  The lattice points are represented by dots and the relevant vectors are circled.  The Voronoi cell $\vor(\Lambda)$ is the shaded region and the packing radius $\rho$ and corresponding sphere packing (circles) are depicted.
}     
		\label{lattices:fig:vorregion}   
\end{figure}

\section{Finding a closest lattice point by a series of relevant vectors} \label{sec:iterative-slicer}

Let $\Lambda$ be a lattice in $\reals^m$ and let $y \in \reals^m$. A simple method to compute a lattice point $x \in \Lambda$ closest to $y$ is as follows.  Let $x_0$ be some lattice point from $\Lambda$, for example the origin.  Consider the following iteration,
\begin{align}
x_{k+1} &= x_k + v_k \nonumber \\
v_k &= \arg\min_{ v \in \relevant(\Lambda) \cup \{\zerobf\} } \|y - x_k - v \|, \label{eq:relvectsminimslicer}
\end{align} 
where $\relevant(\Lambda) \cup \{\zerobf\}$ is the set of relevant vectors of $\Lambda$ including the origin.  The minimum over $\relevant(\Lambda) \cup \{\zerobf\}$ may not be unique, that is, there may be multiple vectors from $\relevant(\Lambda) \cup \{\zerobf\}$ that are closest to $y - x_k$.  In this case, any one of the minimisers may be chosen.  The results that we will describe do not depend on this choice. We make the following straightforward propositions.

\begin{proposition}\label{obs:1}
At the $k$th iteration either $x_k$ is a closest lattice point to $y$ or $\|y - x_k\| > \| y - x_{k+1} \|$.
\end{proposition}
\begin{proof}
If $x_k$ is a closest lattice point to $y$ then $\|y - x_k\| \leq \| y - x_{k+1} \|$ by definition.  On the other hand if $x_k$ is not a closest lattice point to $y$ we have $y - x_k \notin \vor(\Lambda)$ and from~\eqref{eq:relvectnearpointieq} there exists a relevant vector $v$ such that
\[
0 > \dotprod{v}{v} - 2\dotprod{(y - x_k)}{v}.
\]
Adding $\|y - x_k\|^2$ to both sides of this inequality gives
\begin{align*}
\|y - x_k\|^2 &> \dotprod{v}{v} - 2\dotprod{(y - x_k)}{v} + \|y - x_k\|^2 \\
&= \|y - x_k - v\|^2 \\
&\geq \arg\min_{ v \in \relevant(\Lambda) \cup \{\zerobf\}}\|y - x_k - v \|^2 \\
&= \|y - x_k - v_k \|^2 \\
&= \|y - x_{k+1}\|^2. 
\end{align*}
\end{proof} 


 \begin{proposition}\label{obs:2}
 There is a finite number $K$ such that $x_K, x_{K+1}, x_{K+2}, \dots$ are all closest points to $y$.
 \end{proposition}
 \begin{proof}
Suppose no such finite $K$ exists, then
\[
\|y - x_0\| >  \|y - x_1\| > \|y - x_2\| > \dots
\]
and so $x_0,x_1,\dots$ is an infinite sequence of distinct (due to the strict inequality) lattice points  all contained inside an $n$-dimensional hypersphere of radius $r = \|y - x_0\|$ centered at $y$.  This is a contradiction because, if $\rho$ is the packing radius of the lattice, then less than $\ceil{r/\rho}^n$ lattice points lie inside this sphere by Proposition~\ref{eq:latticepointsinsphere}. 
\end{proof}
 
Proposition~\ref{obs:2} above asserts that after some finite number $K$ of iterations the procedure arrives at $x_K$, a closest lattice point to $y$.  Using Proposition~\ref{obs:1} we can detect that $x_K$ is a closest lattice point by checking whether $\|y - x_K\| \leq \| y - x_{K+1} \|$.
This simple iterative approach to compute a closest lattice point is related to what is called the \emph{iterative slicer}~\cite{Shalvi_iterativeslicer_2009}.  Micciancio~\cite{Micciancio09adeterministic} describes a related, but more sophisticated, iterative algorithm that can compute a closest lattice point in a number of operations that grows exponentially as $O(2^{2 n})$.  This single exponential growth in complexity is the best known.  

Two factors contribute to the computational complexity of this iterative approach to compute a closest lattice point.  The first factor is computing the minimum over the set $\relevant(\Lambda) \cup \{\zerobf\}$ in~\eqref{eq:relvectsminimslicer}.  In general a lattice can have as many as $2^{n+1}-2$ relevant vectors so computing a minimiser directly can require a number of operations that grows exponentially with $n$.  To add to this it is often the case that the set of relevant vectors $\relevant(\Lambda)$ must be stored in memory so the algorithm can require an amount of memory that grows exponentially with $n$~\cite[Sec.~6]{Micciancio09adeterministic}\cite{Shalvi_iterativeslicer_2009}.  We will show that for a lattice of Voronoi's first kind the set of relevant vectors has a compact representation in terms of what is called its \emph{obtuse superbasis}.  To store the obtuse superbasis requires an amount of memory of order $O(n^2)$ in the worst case.  We also show that for a lattice of Voronoi's first kind the minimisation over $\relevant(\Lambda) \cup \{\zerobf\}$ in~\eqref{eq:relvectsminimslicer} can be solved efficiently by computing a minimum cut in an undirected flow network.  Using known algorithms a minimiser can be computed in $O(n^3)$ operations~\cite{Goldberg:1986:NAM:12130.12144,EdmondsKarp_max_flow,Cormen2001}. 

The other factor affecting the complexity is the number of iterations required before the algorithm arrives at a closest lattice point, that is, the size of $K$.  Proposition~\ref{eq:latticepointsinsphere} suggests that this number might be as large as $\ceil{r/\rho}^n$ where $r = \|y - x_0\|^2$ and $\rho$ is the packing radius of the lattice.  Thus, the number of iterations required might grow exponentially with $n$.  The number of iterations required depends on the lattice point that starts the iteration $x_0$.  It is helpful for $x_0$ to be, in some sense, a close approximation of the closest lattice point $x_K$.  Unfortunately, computing close approximations of a closest lattice point is known to be computationally difficult~\cite{feige_inapproximability_2004}.  We will show that for a lattice of Voronoi's first kind a simple and easy to compute choice for $x_0$ ensures that a closest lattice point is reached in at most $n$ iterations and so $K \leq n$.  Combining this with the fact that each iteration of the algorithm requires $O(n^3)$ operations results in an algorithm that requires $O(n^4)$ operations to compute a closest point in a lattice of Voronoi's first kind.

\begin{figure}[tp] 
	\centering      
		\includegraphics{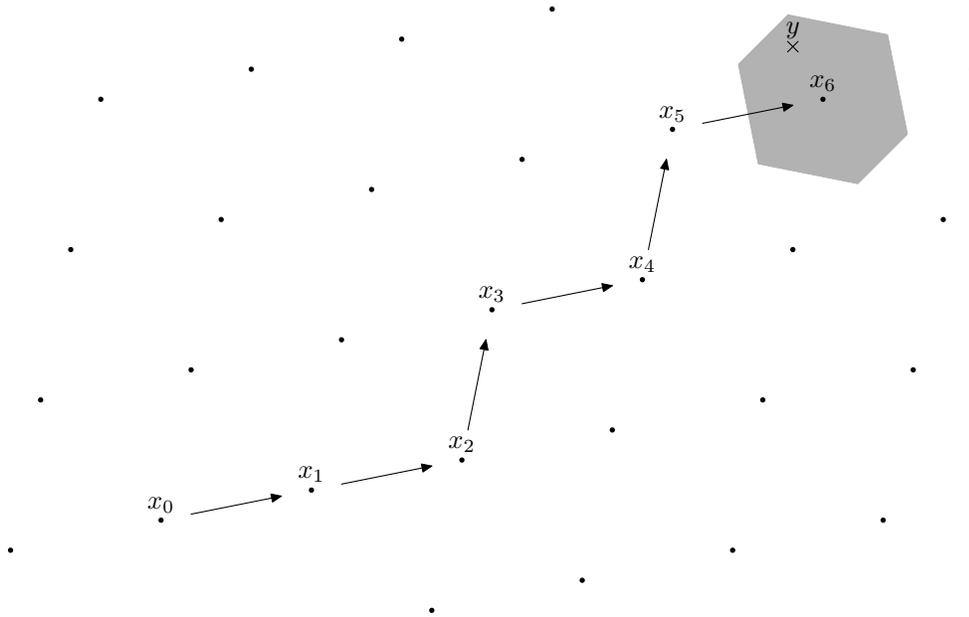} 
		\caption{Example of the iterative procedure described in~\eqref{eq:relvectsminimslicer} to compute a closest lattice point to $y = (4,3.5)$ (marked with a cross) in the $2$-dimensional lattice generated by basis vectors $(2,0.4)$ and $(0.4,2)$.  The initial lattice point for the iteration is $x_0 = (-4.4,-2.8)$.  The shaded region is the Voronoi cell surrounding the closest lattice point $x_6 = (4.4,2.8)$.}       
		\label{lattices:fig:iterativeexample} 
\end{figure}

\section{Lattices of Voronoi's first kind} \label{sec:latt-voron-first}

An $n$-dimensional lattice $\Lambda$ is said to be of \emph{Voronoi's first kind} if it has what is called an \emph{obtuse superbasis}~\cite{ConwaySloane1992_voronoi_lattice_3d_obtuse_superbases}.  That is, there exists a set of $n+1$ vectors $b_1,\dots,b_{n+1}$ such that $b_1,\dots,b_n$ are a basis for $\Lambda$,
\begin{equation}\label{eq:superbasecond}
b_1 + b_2 \dots + b_{n+1} = 0
\end{equation}
(the \emph{superbasis} condition), and the inner products satisfy
\begin{equation}\label{eq:obtusecond}
q_{ij} = b_i \cdot b_j \leq 0, \qquad \text{for} \qquad i,j = 1,\dots,n+1, i \neq j
\end{equation}
(the \emph{obtuse} condition).  The $q_{ij}$ are called the \emph{Selling parameters}~\cite{Selling1874}.  It is known that all lattices in dimensions less than $4$ are of Voronoi's first kind~\cite{ConwaySloane1992_voronoi_lattice_3d_obtuse_superbases}.  An interesting property of lattices of Voronoi's first kind is that their relevant vectors have a straightforward description.

\begin{theorem} \label{thm:revvecssuperbase} (Conway and Sloane~\cite[Theorem~3]{ConwaySloane1992_voronoi_lattice_3d_obtuse_superbases})
The relevant vectors of $\Lambda$ are of the form,
\[
\sum_{i \in I} b_i
\]
where $I$ is a strict subset of $\{1, 2, \dots, n+1\}$ that is not empty, i.e., $I \subset \{1, 2, \dots, n+1\}$ and $I \neq \emptyset$.
\end{theorem}  
 



Classical examples of lattices of Voronoi's first kind are the $n$ dimensional root lattice $A_n$ and its dual lattice $A_n^*$~\cite{SPLAG}.  
For $A_n$ and $A_n^*$ there exist efficient algorithms that can compute a closest lattice point in $O(n)$ operations~\cite{McKilliam2009CoxeterLattices,Conway1982FastQuantDec}. For this reason we do not recommend using the algorithm described in this paper for $A_n$ and $A_n^*$.  
The fast algorithms for $A_n$ and $A_n^*$ rely of the special structure of these lattices and are not applicable to other lattices.  In contrast, the algorithm we describe here works for all lattices of Voronoi's first kind.  Questions that arise are: how ``large'' (in some sense) is the set of lattices of Voronoi's first kind?  Are there lattices of Voronoi's first kind that are useful in applications such as coding, quantisation, or signal processing?  We discuss these questions in Section~\ref{sec:discussion}.  We now focus on the problem of computing a closest lattice point in a lattice of Voronoi's first kind. 


\section{A series of relevant vectors from a lattice of Voronoi's first kind}\label{sec:seri-relev-vect}

We are interested in solving the closest lattice point problem for lattices of Voronoi's first kind.  Let $\Lambda \subset \reals^m$ be an $n$ dimensional lattice of Voronoi's first kind with obtuse superbasis $b_1,\dots,b_{n+1}$ and let $y \in \reals^m$.  We want to find $n$ integers $w_1,\dots,w_n$ that minimise
\[
\| y - \sum_{i=1}^n b_i w_i \|^2.
\]
We can equivalently find $n+1$ integers $w_1,\dots,w_{n+1}$ that minimise
\[
\| y - \sum_{i=1}^{n+1} b_i w_i \|^2.
\]
The iterative procedure described in~\eqref{eq:relvectsminimslicer} will be used to do this.  In what follows it is assumed that $y$ lies in the space spanned by the basis vectors $b_1,\dots,b_{n}$.  This assumption is without loss of generality because $x$ is a closest lattice point to $y$ if and only if $x$ is a the closest lattice point to the orthogonal projection of $y$ into the space spanned by $b_1,\dots,b_{n}$.  Let
\begin{equation}\label{eq:matrxBobtusebasis}
B = (b_1\,\,b_2\,\,\dots\,\,b_{n+1})
\end{equation}
be the $n+1$ by $n+1$ matrix with columns given by $b_1,\dots,b_{n+1}$ and let $z \in \reals^{n+1}$ be a column vector such that $y = Bz$.  We now want to find a column vector $w = (w_1,\dots,w_{n+1})^\prime$ of integers such that
\begin{equation}\label{eq:tominimise}
\| B(z  -  w) \|^2
\end{equation}
is minimised.  Define the column vector $u_0 = \floor{z}$ where $\floor{\cdot}$ denotes the largest integer less than or equal to its argument and operates on vectors elementwise. In view of Theorem~\ref{thm:revvecssuperbase} the iterative procedure~\eqref{eq:relvectsminimslicer} to compute a closest lattice point can be written in the form
\begin{align}
x_{k+1} &= B u_{k+1} \label{eq:xseqfirsttype}  \\
u_{k+1} &= u_k + t_k \nonumber \\
t_k &= \arg\min_{t \in \{0,1\}^{n+1}}\| B(z - u_k - t) \|^2, \label{eq:pvecmin}
\end{align}
where $\{0,1\}^{n+1}$ denotes the set of column vectors of length $n+1$ with elements equal to zero or one.  The procedure is initialised at the lattice point $x_0 = Bu_0 = B\floor{z}$.  This choice of initial lattice point is important.  In Section~\ref{sec:comp-clos-relev} we show how minimisation over $\{0,1\}^{n+1}$ in~\eqref{eq:pvecmin} can be computed efficiently in $O(n^3)$ operations by computing a minimum cut in an undirected flow network.  The minimiser may not be unique corresponding with the existence of multiple minimum cuts.  In this case any one of the minimisers may be chosen.  Our results do not depend on this choice.   In the remainder of this section we prove that this iterative procedure results in a closest lattice point after at most $n$ iterations.  That is, we show that there exists a positive integer $K \leq n$ such that $x_K$ is a closest lattice point to $y = Bz$.

\newcommand{\rng}{\operatorname{rng}}
\newcommand{\subrng}{\operatorname{subr}}
\newcommand{\decrng}{\operatorname{decrng}}

It is necessary to introduce some notation.  For $S$ a subset of indices $\{1,\dots,n+1\}$ let $\onebf_S$ denote the column vector of length $n+1$ with $i$th element equal to one if $i \in S$ and zero otherwise. 
For $S \subseteq \{1,\dots,n+1\}$ and $p \in \reals^{n+1}$ we define the function
\[
\Phi(S, p) = \sum_{i \in S}\sum_{j \notin S}q_{ij}(1 + 2p_i - 2p_j)
\]
where $q_{ij} = \dotprod{b_i}{b_j}$ are the Selling parameters from~\eqref{eq:obtusecond}.  We denote by $\bar{S}$ the complement of the set of indices $S$, that is $\bar{S} = \{ i \in \{1,\dots,n+1\} \mid i \notin S\}$. 

\begin{lemma}\label{lem:decSellings}
Let $p \in \reals^{n+1}$ and let $S$ and $T$ be subsets of the indices of $p$.  The following equalities hold:
\begin{enumerate}
\item  $\|Bp\|^2 - \|B(p + \onebf_S)\|^2 = \Phi(S, p)$, \label{eq:lem:decSellingsinc}
\item  $\|Bp\|^2 - \|B(p - \onebf_S)\|^2 = \Phi(\bar{S}, p)$, \label{eq:lem:decSellingsdec}
\item  ${\displaystyle \|Bp\|^2 - \|B(p + \onebf_S - \onebf_T)\|^2 = \Phi(S, p) + \Phi(\bar{T},p) + \sum_{i\in S}\sum_{j\in T} q_{ij} }$.  \label{eq:lem:decSellingsincdec}
\end{enumerate}
\end{lemma}
\begin{proof}
Part~\ref{eq:lem:decSellingsincdec} follows immediately from parts~\ref{eq:lem:decSellingsinc} and~\ref{eq:lem:decSellingsdec} because
\begin{align*}
\|Bp\|^2 - &\|B(p + \onebf_S - \onebf_T)\|^2 \\
&= \|Bp\|^2-\|B(p + \onebf_S)\|^2 +  \|Bp\|^2-\|B(p - \onebf_T)\|^2 + \sum_{i\in S}\sum_{j\in T} q_{ij}.
\end{align*}
We give a proof for part~\ref{eq:lem:decSellingsinc}.  The proof for part~\ref{eq:lem:decSellingsdec} is similar.  
Put $Q = B^\prime B$ where superscript $^\prime$ indicates the vector or matrix transpose.  The $n+1$ by $n+1$ matrix $Q$ has elements given by the Selling parameters, that is, $Q_{ij} = q_{ij}=\dotprod{b_i}{b_j}$.  Denote by $\onebf$ the column vector of length $n+1$ containing all ones.  Now $B\onebf = \sum_{i=1}^{n+1}b_i = \zerobf$ as a result of the superbasis condition~\eqref{eq:superbasecond} and so $Q\onebf = \zerobf$.  Since $\onebf_S = \onebf - \onebf_{\bar{S}}$ it follows that $Q\onebf_S = -Q\onebf_{\bar{S}}$.  With $\circ$ the elementwise vector product, i.e., the Schur or Hadamard product, we have
 \begin{align*}
 \|Bp\|^2 - \|B(p + \onebf_S)\|^2 &= -\onebf^\prime_S Q \onebf_S - 2 p^\prime Q \onebf_S \\
 &= \onebf^\prime_S Q \onebf_{\bar{S}} - 2 p^\prime Q \onebf_S \\
&= \onebf^\prime_S Q \onebf_{\bar{S}} - 2 (p\circ\onebf_{\bar{S}})^\prime Q \onebf_S - 2 (p\circ\onebf_S)^\prime Q \onebf_S \\
&= \onebf^\prime_S Q \onebf_{\bar{S}} - 2 (p\circ\onebf_{\bar{S}})^\prime Q \onebf_{S} + 2 (p\circ\onebf_S)^\prime Q \onebf_{\bar{S}}
 \end{align*}
which is precisely $\Phi(S,p)$.
\end{proof}


Denote by $\min(p)$ and $\max(p)$ the minimum and maximum values obtained by the elements of the vector $p$ and define the function
\[
\rng(p) = \max(p) - \min(p).
\] 
Observe that $\rng(p)$ cannot be negative and that if $\rng(p) = 0$ then all of the elements of $p$ are equal.  We define the function $\subrng(p)$ to return the largest subset, say $S$, of the indices of $p$ such that $\min\{p_i, i \in S\} - \max\{p_i, i \notin S\} \geq 2.$  If no such subset exists then $\subrng(p)$ is the empty set $\emptyset$.  For example, 
\[
\subrng(2,-1,4) = \{1,3\}, \qquad \subrng(2,1,3) = \emptyset, \quad  \subrng(1,3,1) = \{2\}.
\]  
To make the definition of $\subrng$ clear we give the following alternative and equivalent definition.  Let $p \in \reals^n$ and let $\sigma$ be the permutation of the indices $\{1,\dots,n\}$ that puts the elements of $p$ in ascending order, that is
\[
p_{\sigma(1)} \leq p_{\sigma(2)} \leq \dots \leq p_{\sigma(n)}.
\]  
Let $T$ be the smallest integer from $\{2,\dots,n\}$ such that $p_{\sigma(T)} - p_{\sigma(T-1)} \geq 2$.  If no such integer $T$ exists then $\subrng(p) = \emptyset$.  Otherwise 
\[
\subrng(p) =  \{ \sigma(T), \sigma(T+1), \dots, \sigma(n) \}.
\]
The following straightforward property of $\subrng$ will be useful.

\begin{proposition}\label{prop:subrrngsmall}
Let $p \in \ints^{n+1}$.  If $\subrng(p) = \emptyset$ then $\rng(p) \leq n$.
\end{proposition}
\begin{proof}
Let $\sigma$ be the permutation of the indices $\{1,\dots,n+1\}$ that puts the elements of $p$ in ascending order.  Because $\subrng(p) = \emptyset$ and because the elements of $p$ are integers we have $p_{\sigma(i+1)} \leq p_{\sigma(i)} + 1$ for all $i=1,\dots,n$.  It follows that
\[
p_{\sigma(n+1)} \leq p_{\sigma(n)} + 1 \leq p_{\sigma(n-1)} + 2 \leq \dots \leq p_{\sigma(1)} + n.
\]
and so $\rng(p) = p_{\sigma(n+1)} - p_{\sigma(1)} \leq n$.
\end{proof}

Finally we define the function
\[
\decrng(p) = p -  \onebf_{\subrng(p)}
\]
that decrements those elements from $p$ with indices from $\subrng(p)$.  If $\subrng(p) = \emptyset$, then $\decrng(p) = p$, that is, $\decrng$ does not modify $p$.  On the other hand, if $\subrng(p) \neq \emptyset$ then
\[
\rng\big(\decrng(p)\big) = \rng(p) - 1
\]
because $\subrng(p)$ contains all those indices $i$ such that $p_i = \max(p)$.  By repeatedly applying $\decrng$ to a vector one eventually obtains a vector for which further application of $\decrng$ has no effect.  For example,
\begin{align*}
\decrng(2,-1,4) &= (2,-1,4) - \onebf_{\subrng(2,-1,4)} = (2,-1,4) - \onebf_{\{1,3\}} = (1,-1,3) \\ 
\decrng(1,-1,3) &= (1,-1,3) - \onebf_{\{1,3\}} = (0,-1,2) \\ 
\decrng(0,-1,2) &= (0,-1,2) - \onebf_{\{3\}} = (0,-1,1) \\ 
\decrng(0,-1,1) &= (0,-1,1) - \onebf_{\emptyset} = (0,-1,1).
\end{align*}
This will be a useful property so we state it formally in the following proposition.

\begin{proposition} \label{lem:repeatappdecrange} Let $p \in \reals^{n+1}$ and define the infinite sequence $d_0,d_1,d_2,\dots$ of vectors according to $d_0=p$ and $d_{k+1} = \decrng(d_k)$.  There is a finite integer $T$ such that $d_T=d_{T+1}=d_{T+2}=\dots$.
\end{proposition}
\begin{proof}
Assume that no such $T$ exists.  Then $\decrng(d_k) \neq d_k$ for all positive integers $k$ and so 
\[
\rng(d_{k}) = \rng(d_{k-1}) - 1 = \rng(d_{k-2}) - 2 = \dots = \rng(p) - k.
\]  
Choosing $k > \rng(p)$ we have $\rng(d_{k}) < 0$ contradicting that $\rng(d_k)$ is nonegative.
\end{proof}

We are now ready to study properties of a closest lattice point in a lattice of Voronoi's first kind.

\begin{lemma}\label{lem:decrngpreservesclosestpoints}
If $v \in \ints^{n+1}$ such that $B(\floor{z} + v)$ is a closest lattice point to $y = Bz$, then $B\big(\floor{z} + \decrng(v)\big)$ is also a closest lattice point to $y$.
\end{lemma}
\begin{proof}
The lemma is trivial if $\subrng(v) = \emptyset$ so that $\decrng(v) = v$.  It remains to prove the lemma when $\subrng(v) \neq \emptyset$.  In this case put $S = \subrng(v)$ and put 
\[
u = \decrng(v) = v - \onebf_S.
\] 
Let $\zeta = z - \floor{z}$ be the column vector containing the fractional parts of the elements of $z$.  We have $\zeta - u = \zeta - v + \onebf_S$.  Applying part~\ref{eq:lem:decSellingsinc} of Lemma~\ref{lem:decSellings} with $p = \zeta - v$ we obtain
\begin{align}
\|B(\zeta - v)\|^2 - \|B(\zeta - u)\|^2 &= \Phi(S, \zeta-v) \nonumber \\
&= \sum_{i \in S}\sum_{j \notin S}q_{ij}\big(1 + 2(\zeta_i-\zeta_j) - 2(v_i - v_j)\big). \label{eq:sumsumBzBu}
\end{align}
Observe that $\zeta_i =  z_i - \floor{z_i} \in [0,1)$ for all $i=1,\dots,n+1$ and so $-1 < \zeta_i-\zeta_j < 1$ for all $i,j=1,\dots,n+1$.  Also, for $i \in S$ and $j\notin S$ we have 
\[
v_i - v_j \geq \min\{ v_i, i \in S\} - \max\{v_j, j \notin S\} \geq 2
\]
by definition of $\subrng(v) = S$.  Thus,
\[
1 + 2(\zeta_i-\zeta_j) - 2(v_i - v_j) < 1 + 2 - 4 = -1 < 0 \qquad \text{for $i \in S$ and $j \notin S$}.
\]
Substituting this inequality into~\eqref{eq:sumsumBzBu} and using that $q_{ij} \leq 0$ for $i \neq j$ (the obtuse condition~\eqref{eq:obtusecond}) we find that
\[
\|B(z - \floor{z} - v)\|^2 - \|B(z - \floor{z} - u)\|^2 \geq 0.
\]
It follows that $B(\floor{z} + u) = B\big(\floor{z} + \decrng(v)\big)$ is a closest lattice point to $y = Bz$ whenever $B(\floor{z} + v)$ is.
\end{proof}

\begin{lemma}\label{lem:roundzclose}
There exists a closest lattice point to $y = Bz$ in the form $B(\floor{z} + v)$ where $v \in \ints^{n+1}$ with $\rng(v) \leq n$.
\end{lemma}
\begin{proof}
Let $d_0 \in \ints^{n+1}$ be such that $B(\floor{z} + d_0)$ is a closest lattice point to $y$. Define the sequence of vectors $d_0,d_1,\dots$ from $\ints^{n+1}$ according to the recursion $d_{k+1} = \decrng(d_k)$.  It follows from Lemma~\ref{lem:decrngpreservesclosestpoints} that $B(\floor{z} + d_{k})$ is a closest lattice point for all positive integers $k$.  By Proposition~\ref{lem:repeatappdecrange} there is a finite $T$ such that 
\[
d_{T+1}=d_T=\decrng(d_T).
\]  
Thus $\subrng(d_T) = \emptyset$ and $\rng(d_T) \leq n$ by Proposition~\ref{prop:subrrngsmall}.  The proof follows with $v = d_T$.   
\end{proof}

Let $\ell$ be a nonegative integer.  We say that a lattice point $x$ is $\ell$-\emph{close} to $y$ if there exists a $v \in \ints^{n+1}$ with $\rng(v) = \ell$ such that $x + Bv$ is a closest lattice point to $y$.  Lemma~\ref{lem:roundzclose} asserts that the lattice point $x_0 = B\floor{z}$ that initialises the iterative procedure~\eqref{eq:xseqfirsttype} is $K$-close to $y$ where $K \leq n$.  From Lemma~\ref{lem:rngdecreases} stated below it will follow that if the lattice point $x_k$ obtained on the $k$th iteration of the procedure is $\ell$-close, then the lattice point $x_{k+1}$ obtained on the next iteration is $(\ell-1)$-close.  Since $x_0$ is $K$-close it will then follow that after $K \leq n$ iterations the lattice point $x_K$ is $0$-close.  At this stage it is guaranteed that $x_{K}$ itself is a closest lattice point to $y$.  This is shown in the following lemma.  

\begin{lemma}\label{lem:rngzeroclosestpoint}
If $x$ is a lattice point that is $0$-close to $y$, then $x$ is a closest lattice point to $y$.
\end{lemma}
\begin{proof}
Because $x$ is $0$-close there exists a $v \in \ints^{n+1}$ with $\rng(v) = 0$ such that $x + Bv$ is a closest lattice point to $y$.  Because $\rng(v) = 0$ all elements from $v$ are identical, that is, $v_1=v_2=\dots=v_{n+1}$.  In this case $Bv = \sum_{i=1}^{n+1} v_n b_n = v_1\sum_{i=1}^{n+1}b_n = 0$
as a result of the superbasis condition~\eqref{eq:superbasecond}.  Thus $x = x + Bv$ is a closest point to $y$. 
\end{proof}

Before giving the proof of Lemma~\ref{lem:rngdecreases} we require the following simple result.

\begin{lemma}\label{eq:integergreaterless}
Let $h \in \{0,1\}^{n+1}$ and $v \in \ints^{n+1}$.  Suppose that
$h_i = 0$ whenever $v_i = \min(v)$ and that $h_i = 1$ whenever $v_i = \max(v)$.  Then
\[
h_i - h_j \leq v_i - v_j
\]
when either $v_i = \max(v)$ or $v_j = \min(v)$.
\end{lemma}
\begin{proof}
If $v_i = \max(v)$ then $h_i = 1$ and we need only show that $1-h_j \leq \max(v) - v_j$ for all $j$.  If $v_j = \max(v)$ then $h_j = 1$ and the results holds since $1 - h_j = 0 = \min(v) - v_j$.  Otherwise if $v_j < \max(v)$ then $1-h_j \leq 1 \leq \max(v) - v_j$ because $\max(v)$ and $v_j$ are integers.

Now, if $v_j = \min(v)$ then $h_j = 0$ and we need only show that $h_i \leq v_i - \min(v)$ for all $i$.  If $v_i = \min(v)$ then $h_i = 0 = v_i - \min(v)$ and the results holds.  Otherwise if $v_i > \min(v)$ then $h_i \leq 1 \leq v_i - \min(v)$ because $\min(v)$ and $v_i$ are integers.
\end{proof}

\begin{lemma}\label{lem:rngdecreases}
Let $Bu$ with $u \in \ints^{n+1}$ be a lattice point that is $\ell$-close to $y = Bz$ where $\ell > 0$.  Let $g \in \{0,1\}^{n+1}$ be such that
\begin{equation}\label{eq:qismin}
\|B(z - u - g)\|^2 = \min_{t \in \{0,1\}^{n+1}}\|B(z - u - t)\|^2.
\end{equation}
The lattice point $B(u+g)$ is $(\ell-1)$-close to $y$.
\end{lemma}
\begin{proof}
Because $Bu$ is $\ell$-close to $y$ there exists $v \in \ints^{n+1}$ with $\rng(v) = \ell$ such that $B(u+v)$ is a closest lattice point to $y$.  Define subsets of indices
\[
S = \{i \mid g_i = 0, v_i = \max(v) \}, \qquad T = \{i \mid g_i = 1, v_i = \min(v) \}, 
\]
and put $h = g + \onebf_S - \onebf_T \in \{0,1\}^{n+1}$ and $w = v - h$.  Observe that $h_i = 0$ whenever $v_i = \min(v)$ and so $\min(w) = \min(v-h) = \min(v)$.  Also, $h_i = 1$ whenever $v_i = \max(v)$ and so $\max(w) = \max(v-h) = \max(v) - 1$.  Thus,
\[
\rng(w) = \max(v) - 1 - \min(v) = \rng(v) - 1 = \ell - 1.
\]
The lemma will follow if we show that $B(u+g+w)$ is a closest lattice point to $y$ since then the lattice point $B(u+g)$ with be $(\ell-1)$-close to $y$.  The proof is by contradiction.  Suppose $B(u+g+w)$ is not a closest point to $y$, that is, suppose
\[
\|B(z-u-g-w)\|^2 > |B(z-u-v)\|^2.
\]
Putting $p = z-u-v$ we have
\[
\|B(p+\onebf_S-\onebf_T)\|^2 > \|Bp\|^2.
\]
By part~\ref{eq:lem:decSellingsincdec} of Lemma~\ref{lem:decSellings} we obtain
\begin{equation}\label{eq:BpBpdecincineq}
\|Bp\|^2 - \|B(p+\onebf_S-\onebf_T)\|^2 = \Phi(S,p) + \Phi(\bar{T}, p) + \sum_{i\in S}\sum_{j \in S} q_{ij} < 0.
\end{equation}

As stated already $h_i = 0$ whenever $v_i = \min(v)$ and $h_i = 1$ whenever $v_i = \max(v)$.  It follows from Lemma~\ref{eq:integergreaterless} that
\begin{equation}\label{eq:hijvijieq}
h_{i} - h_j \leq v_{i} - v_j
\end{equation}
when either $v_i = \max(v)$ or $v_j = \min(v)$.  Since $v_i = \max(v)$ for $i \in S$ and $v_j = \min(v)$ for $j \in T$ the inequality~\eqref{eq:hijvijieq} holds when either $i \in S$ or $j \in T$.

Put $r = z - u - h$.  By~\eqref{eq:hijvijieq} we have $r_i - r_j \geq p_{i} - p_{j}$ when either $i \in S$ or $j \in T$.  Now, since $q_{ij} \leq 0$ for $i \neq j$,
\[
\Phi(S,r) = \sum_{i \in S}\sum_{j \notin S}q_{ij}(1 + 2r_i - 2r_j) \leq \sum_{i \in S}\sum_{j \notin S}q_{ij}(1 + 2p_i - 2p_j) = \Phi(S,p),
\]
and
\[
\Phi(\bar{T},r) = \sum_{i \notin T}\sum_{j \in T}q_{ij}(1 + 2r_i - 2r_j) \leq \sum_{i \notin T}\sum_{j \in T}q_{ij}(1 + 2p_i - 2p_j) = \Phi(\bar{T},p).
\]
Using part~\ref{eq:lem:decSellingsincdec} of Lemma~\ref{lem:decSellings} again,
\begin{align*}
\|B(z-u-h)\|^2 - \|B(z-u-g)\|^2 &= \|Br\|^2 - \|B(r+\onebf_S - \onebf_T)\|^2 \\
&= \Phi(S,r) + \Phi(\bar{T},r) + \sum_{i\in S}\sum_{j \in S} q_{ij} \\
&\leq \Phi(S,p) + \Phi(\bar{T}, p) + \sum_{i\in S}\sum_{j \in S} q_{ij} < 0 
\end{align*}
as a result of~\eqref{eq:BpBpdecincineq}.  However, $h \in \{0,1\}^{n+1}$ and so this implies
\[
\|B(z-u-g)\|^2 > \|B(z-u-h)\|^2 \geq \min_{t \in \{0,1\}^{n+1}}\|B(z - u - t)\|^2
\]
contradicting~\eqref{eq:qismin}.  Thus, our original supposition is false and $B(u+g+w)$ is a closest lattice point to $y$. Because $\rng(w) = \ell-1$ the lattice point $B(u+g)$ is $(\ell-1)$-close to $y$.
 \end{proof}

The next theorem asserts that the iterative procedure~\eqref{eq:xseqfirsttype} converges to a closest lattice point in $K \leq n$ iterations.  This is the primary result of this section.

\begin{theorem}
Let $x_0,x_1,\dots$ be the sequence of lattice points given by the iterative procedure~\eqref{eq:xseqfirsttype}.  There exists $K \leq n$ such that $x_K$ is a closest lattice point to $y = Bz$.
\end{theorem}
\begin{proof}
Let $x_k = B u_k$ be the lattice point obtained on the $k$th iteration of the procedure.  Suppose that $x_k$ is $\ell$-close to $y=Bz$ with $\ell > 0$.  The procedure computes $t_{k} \in \{0,1\}^{n+1}$ satisfying
\[
\|B(z - u_k - t_{k})\|^2 = \min_{t \in \{0,1\}^{n+1}}\|B(z - u_k - t)\|^2
\]
and puts $x_{k+1} = B(u_k + t_k)$.  It follows from Lemma~\ref{lem:rngdecreases} that $x_{k+1}$ is $(\ell-1)$-close to $y$.  By Lemma~\ref{lem:roundzclose} the lattice point that initialises the procedure $x_0 = B \floor{z}$ is $K$-close to $y$ where $K \leq n$.  Thus, $x_1$ is $(K-1)$-close, $x_2$ is $(K-2)$-close, and so on until $x_K$ is $0$-close.  That $x_K$ is a closest lattice point to $y$ follows from Lemma~\ref{lem:rngzeroclosestpoint}.
\end{proof}

\section{Computing a closest relevant vector}\label{sec:comp-clos-relev}

In the previous section we showed that the iterative procedure~\eqref{eq:xseqfirsttype} results in a closest lattice point in at most $n$ iterations.  It remains to show that each iteration of the procedure can be computed efficiently.  Specifically, it remains to show that the minimisation over the set of binary vectors $\{0,1\}^{n+1}$ described in~\eqref{eq:pvecmin} can be computed efficiently.  Putting $p = z - u_k$ in~\eqref{eq:pvecmin} we require an efficient method to compute a $t \in \{0,1\}^{n+1}$ such that the binary quadratic form
\[
\| B(p - t) \|^2 = \| \sum_{i=1}^{n+1} b_i (p_i - t_i) \|^2
\]
is minimised.  Expanding this quadratic form gives
\[
\| \sum_{i=1}^{n+1} b_i (p_i - t_i) \|^2 =  \sum_{i=1}^{n+1}\sum_{j=1}^{n+1} q_{ij}p_i p_j -  2\sum_{i=1}^{n+1}\sum_{j=1}^{n+1} q_{ij}p_j t_i + \sum_{i=1}^{n+1}\sum_{j=1}^{n+1} q_{ij} t_i t_j.
\]
The first sum above is independent of $t$ and can be ignored for the purpose of minimisation.  Letting $s_i = \sum_{j=1}^{n+1} q_{ij}p_j$, we can equivalently minimise the binary quadratic form
\begin{equation}\label{eq:quadformnp}
Q(t) = \sum_{i=1}^{n+1} s_i t_i + \sum_{i=1}^{n+1}\sum_{j=1}^{n+1} q_{ij} t_i t_j.
\end{equation}
We will show that a minimiser of $Q(t)$ can be found efficiently be computing a minimum cut in an undirected flow network.
This technique has appeared previously~\cite{Picard_min_cuts_1974,Sankaran_solving_CDMA_mincut_1998,Ulukus_cdma_mincut_1998,Cormen2001} but we include the derivation here so that this paper is self contained.

Let $G$ be an undirected graph with $n+3$ vertices $v_0, \dots, v_{n+2}$ contained in the set $V$ and edges $e_{ij}$ connecting $v_i$ to $v_j$.  To each edge we assign a \emph{weight} $w_{ij} \in \reals$.  The graph is undirected so the weights are symmetric, that is, $w_{ij} = w_{ji}$.  By calling the vertex $v_0$ the \emph{source} and the vertex $v_{n+2}$ the \emph{sink} the graph $G$ is what is called a \emph{flow network}.  The flow network is \emph{undirected} since the weights assigned to each edge are undirected.  A \emph{cut} in the flow network $G$ is a subset $C \subset V$ of vertices with its complement $\bar{C} \subset V$ such that the source vertex $v_0 \in C$ and the sink vertex $v_{n+2} \in \bar{C}$.  

The weight of a cut is
\[
W(C,\bar{C}) = \sum_{i \in I} \sum_{j \in J} w_{ij}, 
\]
where $I = \{ i \mid v_i \in C\}$ and $J = \{j \mid v_j \in \bar{C}\}$.  That is, $W(C,\bar{C})$ is the sum of the weights on the edges crossing from the vertices in $C$ to the vertices in $\bar{C}$.  In what follows we will often drop the argument and write $W$ rather than $W(C,\bar{C})$.  A \emph{minimum cut} is a $C$ and $\bar{C}$ that minimise the weight $W$.  If all of the edge weights $w_{ij}$ for $i \neq j$ are nonnegative, a minimum cut can be computed in 
order $O(n^3)$ arithmetic operations~\cite{Cormen2001,Even_graph_algorithms_1979}.

We require some properties of the weights $w_{ij}$ in relation to $W$.  If the graph is allowed to contain loops, that is, edges from a vertex to itself, then the weight of these edges $w_{ii}$ have no effect on the weight of any cut.  We may choose any values for the $w_{ii}$ without affecting $W$.  We will find it convenient to set $w_{0,0} = w_{n+2,n+2} = 0$.  The remaining $w_{ii}$ we shall specify shortly.  The edge $e_{0,n+2}$ is in every cut.  If a constant is added to the weight of this edge, that is, $w_{0,n+2}$ is replaced by $w_{0,n+2} + c$ then $W$ is replaced by $W + c$ for every $C$ and $\bar{C}$.  In particular, the subsets $C$ and $\bar{C}$ corresponding to a minimum cut are not changed.  We will find it convenient to choose $w_{0,n+2} = w_{n+2,0} = 0$.  

If vertex $v_i$ is in $C$ then edge $e_{i,n+2}$ contributes to the weight of the cut.  If $v_i \notin C$, i.e., $v_i \in \bar{C}$, then edge $e_{0,i}$ contributes to the weight of the cut.  So, either $e_{0,i}$ or $e_{i,n+2}$ \emph{but not both} contribute to every cut.  If a constant, say $c$, is added to the weights of these edges, that is, $w_{0,i}$ and $w_{i,n+2}$ are replaced by $w_{0,i} + c$ and $w_{i,n+2} + c$, then $W$ is replaced by $W + c$ for every $C$ and $\bar{C}$.  The $C$ and $\bar{C}$ corresponding to a minimum cut are unchanged.  In this way, the minimum cut is only affected by the differences 
\[
d_i = w_{i,n+2} - w_{0,i}
\]
for each $i$ and not the specific values of the weights $w_{i,n+2}$ and $w_{0,i}$.  


We now show how $W(C,\bar{C})$ can be represented as a binary quadratic form.  Put $t_0 = 1$ and $t_{n+2} = 0$ and
\[
t_i = \begin{cases}
1, & i \in C \\
0, & i \in \bar{C}
\end{cases}
\]
for $i = 1,2,\dots,n+1$.  Observe that
\[
t_i(1 - t_j) = \begin{cases}
1, & i \in C, j \in \bar{C} \\
0, & \text{otherwise}.
\end{cases}
\]
The weight can now be written as
\begin{align*}
W(C,\bar{C}) = \sum_{i \in C} \sum_{j \in \bar{C}} w_{ij} = \sum_{i =0}^{n+2} \sum_{j =0}^{n+2} w_{ij} t_i (1 - t_j) = F(t),
\end{align*}
say.  Finding a minimum cut is equivalent to finding the binary vector $t = (t_1, \dots, t_{n+1})$ that minimises $F(t)$.  Write,
\[
F(t) =  \sum_{i=0}^{n+2} \sum_{j =0}^{n+2} w_{ij}t_i - \sum_{i=0}^{n+2} \sum_{j =0}^{n+2} w_{ij} t_it_j.
\]
Letting $k_i = \sum_{j =0}^{n+2} w_{ij}$, and using that $t_0 = 1$ and $t_{n+2} = 0$,
\[
F(t) = \sum_{i=0}^{n+1}k_it_i  - w_{00} - \sum_{i=1}^{n+1} w_{i0} t_i - \sum_{j=1}^{n+1} w_{0j} t_j - \sum_{i=1}^{n+1} \sum_{j =1}^{n+1} w_{ij} t_it_j.
\]
Because $w_{00} = 0$ and $w_{ij} = w_{ji}$ we have
\[
F(t) = k_0 + \sum_{i=1}^{n+1}( k_i  - 2 w_{i0}) t_i  - \sum_{i=1}^{n+1} \sum_{j =1}^{n+1} w_{ij} t_it_j.
\]
The constant term $k_0$ is unimportant for the purpose of minimisation so finding a minimum cut is equivalent to minimising the binary quadratic form
\[
\sum_{i=1}^{n+1}g_i t_i  - \sum_{i=1}^{n+1} \sum_{j =1}^{n+1} w_{ij} t_it_j,
\]
where $g_i = k_i  - 2 w_{i0} = d_i + \sum_{j=1}^{n+1} w_{ij}$.  It only remains to observe the equivalence of this quadratic form and $Q(t)$ from~\eqref{eq:quadformnp} when the weights are assigned to satisfy,
\begin{align*}
&q_{ij} = - w_{ij} \qquad i,j = 1,\dots,n+1 \\
&s_i = g_i = d_i + \sum_{j=1}^{n+1} w_{ij}.
\end{align*}
Because the $q_{ij}$ are nonpositive for $i \neq j$ the weights $w_{ij}$ are nonnegative for all $i \neq j$ with $i,j = 1,\dots,n+1$.  As discussed the value of the weights $w_{ii}$ have no effect on the weight of any cut $W$ so setting $q_{ii} = - w_{ii}$ for  $i = 1,\dots,n+1$ is of no consequence.  Finally the weights $w_{i,n+2}$ and $w_{0,i}$ can be chosen so that both are nonnegative and 
\[
w_{i,n+2} - w_{0,i} = d_i = s_i + \sum_{j=1}^{n+1} q_{ij} = s_i
\]  
because $\sum_{j=1}^{n+1} q_{ij} = 0$ due to the superbase condition~\eqref{eq:superbasecond}.  That is, we choose $w_{i,n+2} = s_i$ and $w_{0,i} = 0$ when $s_i \geq 0$ and $w_{i,n+2}=0$ and $w_{0,i} = -s_i$ when $s_i < 0$.  With these choices, all the weights $w_{ij}$ for $i \neq j$ are nonnegative.  A minimiser of $Q(t)$, and correspondingly a solution of~\eqref{eq:pvecmin} can be computed in $O(n^3)$ operations by computing a minimum cut in the undirected flow network $G$ assigned with these nonnegative weights~\cite{Picard_min_cuts_1974,Sankaran_solving_CDMA_mincut_1998,Ulukus_cdma_mincut_1998,Cormen2001}.  

\section{Discussion}\label{sec:discussion} 

The closest lattice point problem has a number of applications, for example, channel coding and data quantisation~\cite{Conway1983VoronoiCodes,Conway1982VoronoiRegions,Conway1982FastQuantDec,Erex2004_lattice_decoding,Erez2005}.  A significant hurdle in the practical application of lattices as codes or as quantisers is that computing a closest lattice point is computationally difficult in general~\cite{micciancio_hardness_2001}.  The best known general purpose algorithms require a number of operations of order $O(2^{2n})$~\cite{Micciancio09adeterministic}.  In this paper we have focused on the class of lattices of Voronoi's first kind.  We have shown that computing a closest point in a lattice of Voronoi's first kind can be achieved in a comparatively modest number of operations of order $O(n^4)$.  Besides being of theoretical interest, the algorithm has potential for practical application.

A question of immediate interest to communications engineers is: do there exist lattices of Voronoi's first kind that produce good codes or good quantisers?  Since lattices that produce good codes and quantisers often also describe dense~\emph{sphere packings}~\cite{SPLAG}, a related question is: do there exist lattices of Voronoi's first kind that produce dense sphere packings?  These questions do not appear to have trivial answers.  The questions have heightened importance due to the algorithm described in this paper.

It is straightforward to construct an `arbitrary' lattice of Voronoi's first kind.  One approach is to construct the $n+1$ by $n+1$ symmetric matrix $Q = B^\prime B$ with elements $Q_{ij} = q_{ij} = \dotprod{b_{i}}{b_j}$ given by the Selling parameters.  Choose the off diagonal entries of $Q$ to be nonpositive with $q_{ij}=q_{ji}$ and set the diagonal elements $q_{ii} = -\sum_{j=1}^{n+1} q_{ij}$.  
The matrix $Q$ is diagonally dominant, that is, $\abs{q_{ii}} \geq \sum_{j=1}^{n+1} \abs{q_{ij}}$, and so $Q$ is positive semidefinite.  A rank deficient Cholesky decomposition~\cite{Higham90analysisof} can now be used to recover a matrix $B$ such that $B^\prime B  = Q$. The columns of $B$ are vectors of the obtuse superbasis.

A number applications such as phase unwrapping~\cite{Teunissen_GPS_1995,Hassibi_GPS_1998}, single frequency estimation~\cite{McKilliamFrequencyEstimationByPhaseUnwrapping2009}, and related signal processing problems~\cite{McKilliam2007,Clarkson2007,McKilliam2009IndentifiabliltyAliasingPolyphase,Quinn_sparse_noisy_SSP_2012} also require computing a closest lattice point.  In these applications the particular lattice arises from the signal processing problem under consideration.  If that lattice happens to be of Voronoi's first kind then our algorithm can be used.  An example where this occurs is the problem of computing the \emph{sample intrinsic mean} in circular statistics~\cite{McKilliam_mean_dir_est_sq_arc_length2010}.  In this particular problem the lattice $A_n^*$ is involved.  A fast closest point algorithm requiring only $O(n)$ operations exists for $A_n^*$~\cite{McKilliam2009CoxeterLattices,McKilliam2008b} and so the algorithm described in this paper is not needed in this particular case.  However, there many exist other signal processing problems where lattices of Voronoi's first kind arise.

A final remark is that our algorithm assumes that the obtuse superbasis is known in advance.  It is known that all lattices of dimension less than 4 are of Voronoi's first kind and an algorithm exists to recover the obtuse superbasis in this case~\cite{SPLAG}.  Lattices of dimension larger than 4 need not be of Voronoi's first kind.  An interesting question is: given a lattice, is it possible to efficiently decide whether it is of Voronoi's first kind?  A related question is: is it possible to efficiently find an obtuse superbasis if it exists?


\section{Conclusion}

The paper describes an algorithm to compute a closest lattice point in a lattice of Voronoi's first kind when the obtuse superbasis is known~\cite{ConwaySloane1992_voronoi_lattice_3d_obtuse_superbases}.  The algorithm requires $O(n^4)$ operations where $n$ is the dimension of the lattice.  The algorithm iteratively computes a series of relevant vectors that converges to a closest lattice point after at most $n$ terms.   Each relevant vector in the series can be efficiently computed in $O(n^3)$ operations by computing a minimum cut in an undirected flow network.  The algorithm has potential application in communications engineering problems such as coding and quantisation.  An interesting problem for future research is to find lattices of Voronoi's first kind that produce good codes, good quantisers, or dense sphere packings~\cite{SPLAG,Conway1982VoronoiRegions}.


\end{document}